\documentclass{lmcs}
\usepackage{mik-makro,amsmath,mathrsfs,color}
\usepackage{xspace,pdfsync}
\usepackage[all]{xy}

\newcommand*\mso{\textsc{mso}\xspace}
\newcommand{\matrclos}{\mathsf{M}}
\newcommand{\wreaclos}{\mathsf{W}}
\newcommand{\mypic}[1]{
	\begin{center}
		\includegraphics[page=#1,scale=0.4]{pics}
	\end{center}
}

\newcommand{\langclass}{\mathscr L}
\newcommand{\algclass}{\mathscr A}
\newcommand{\pol}[1]{\mathsf{pol}#1}

\newcommand{\dtop}{{\sc dtop}\xspace}
\newcommand{\dtops}{{\sc dtop}s\xspace}

\newcommand{\trees}[1]{\mathsf{trees}#1}

\newcommand{\alg}{\mathbf{A}}
\newcommand{\balg}{\mathbf{B}}

\newcommand{\eqdef}{\stackrel{\text{def}}=}

\newcounter{ourexamplecounter}

\title{Some connections between universal algebra and logics for trees}
\author{Miko{\l}aj Boja\'nczyk and Henryk Michalewski (University of Warsaw)}
\begin{document}
\maketitle
\begin{abstract}
One of the major open problems in automata and logic is the following: is there an algorithm which inputs a regular tree language and decides if the language can be defined in first-order logic? The goal of this paper is to present this problem and similar ones using the language of universal algebra, highlighting potential connections to the structural theory of finite algebras, including Tame Congruence Theory.
\end{abstract}
\section{Introduction}
This paper is dedicated to the memory of Zolt\'an \'Esik, in recognition of  his many contributions to the algebraic theory of languages. Our topic is the following problem, and similar ones:
\begin{itemize}
	\item 	Is there an algorithm which inputs a regular tree language and decides if the language can be defined in first-order logic?
\end{itemize}
 For regular words languages, the answer is positive, as shown by Sch\"utzenberger~\cite{DBLP:journals/iandc/Schutzenberger65a} together with McNaughton-Papert~\cite{mcnaughton}.  Furthermore, the solution in the case of words uses algebra: a regular word language is definable in first-order logic if and only if its syntactic semigroup does not contain a group. This result has been an inspiration for a field called \emph{algebraic language theory}, see~\cite{pinbook}. There is a history of failed attempts to generalise the Sch\"utzenberger-McNaughton-Papert result to trees, see~\cite{bojanczyk-tree-algs} for a discussion.  Remarkably, the attempts to characterise first-order logic (and related logics) for trees have not used the structural theory of finite algebras, e.g.~Tame Congruence Theory~\cite{hobby1988structure}, a theory which has gained importance in theoretical computer science due to its application to  classifying Constraint Satisfaction Problems~\cite{DBLP:conf/dagstuhl/BartoKW17}.
The goal of this paper is to present the questions about tree logics using the language of universal algebra. We hope that this would make it easier for (a) specialists in universal algebra to attack the formal language problems; (b) specialists in formal languages to start using the tools of universal algebra.

This paper is organised as follows.
\begin{itemize}
	\item In Section~\ref{sec:languages}, we describe trees and how logics can define sets of trees. We also describe the main topic of this paper,  the   \emph{definability problem}, which is the following decision problem parametrised by a logic $\langclass$:  decide if a given regular tree language can be defined by some formula of the logic $\langclass$.
	\item In Section~\ref{sec:definability-algebraic}, we show how for many logics of interest, the definability problem can be recast as a question about finite algebras.
	\item In Section~\ref{sec:clones} we give a very brief description of the structural theory of finite algebras, and discuss some preliminary ideas on how it might be used to solve the definability problem.
	\item In Section~\ref{sec:matrix} we draw the connection between the matrix power of finite algebras and reductions via deterministic top-down tree transducers.
	\item In Section~\ref{sec:wreath} we draw the connection between the wreath products of finite algebras and nesting of tree languages.
\end{itemize}

\section{Tree languages and logics defining them}
\label{sec:languages}

\newcommand{\aalg}{\mathbf A}
\newcommand{\clonec}{\mathbf C}
\newcommand{\cloned}{\mathbf D}
There are many variants of trees studied in the formal language literature, including finite and infinite trees, with ranked or unranked alphabets. For an overview of algebraic approaches to these trees, see~\cite{bojanczyk-tree-algs}.  In this paper we talk about finite trees over a ranked alphabet, which is the formalism most closely connected to universal algebra.   Define a \emph{ranked alphabet} to be a finite set $\Sigma$, with each element $a \in \Sigma$ associated an arity in $\set{0,1,\ldots,}$.   Define \emph{tree over $\Sigma$} to be a finite tree where each node is labelled by a label from $\Sigma$ such that the arity of the label is equal to the number of children. We assume that the children are ordered, i.e.~it make sense to talk about the first child, second child, etc. We write $\trees \Sigma$ for the set of trees over $\Sigma$. Here is a picture of a ranked alphabet:
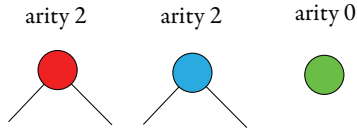
\mypic{1} 
Here is a picture of a tree over the above ranked alphabet, together with the standard tree terminology that we use in this paper:
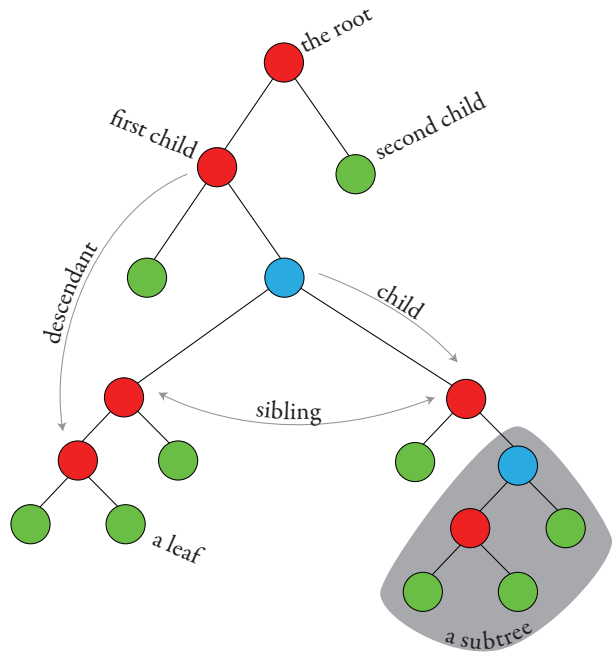
\mypic{2}
\subsection{Algebras.} Define an algebra $\aalg$ to be a set $A$, called the \emph{carrier} of the algebra, together with a set of operations of type $f : A^n \to A$, with possibly different arities $n$. An algebra is called \emph{finite} if its carrier is finite and its set of operations is also finite. We adopt the convention that algebras are written in boldface, e.g.~$\aalg$ or $\balg$, and their respective carriers are denoted using the same letter but not in boldface, e.g.~$A$ or $B$.

\begin{exa}
The Boolean algebra, which we denote by $(2, \lor, \land, \neg)$, has  carrier $\set{0,1}$ and  two binary operations $\set{0,1}^2 \to \set{0,1}$ standing for disjunction and conjunction, as well as one unary operation $\set{0,1} \to \set{0,1}$ standing for negation. If we remove $\neg$ from the set of operations, then the algebra is called the (Boolean) lattice, and if we keep only $\lor$ in the operations then it is called the (Boolean) semi-lattice.
\end{exa}

\begin{exa}\label{ex:trees-as-free-algebra}
	If $\Sigma$ is a ranked alphabet, then $\trees \Sigma$ can be viewed as an algebra, where the carrier is all trees, and there is an operation for every letter $a \in \Sigma$ which combines trees in the obvious way.
\end{exa}

Note that in the definition of algebra above, there are no names for the operations. An alternative would be to consider $\Sigma$-algebras, where $\Sigma$ is some ranked alphabet; in a $\Sigma$-algebra the operations are indexed by letters from $\Sigma$ with corresponding arities. Such algebras are sometimes called \emph{indexed algebras}. Indexed algebras are the more common formalism in the formal language community, see e.g.~\cite{thatcher} which introduces regular tree languages, or the survey books~\cite{tata2007,DBLP:books/others/tree1984}. We use non-indexed algebra here,  to be more consistent with the literature on finite algebras, where non-indexed algebras are more prevalent, e.g.~\cite{hobby1988structure}.

\subsection{Tree languages.}
%
%
%

 A \emph{tree language} over a ranked alphabet $\Sigma$  is defined to be any subset $L \subseteq \trees \Sigma$.  We use algebras to recognise tree languages in the following way.  Define a function from  $\trees \Sigma$ to the carrier   of an algebra $\aalg$ to be a \emph{homomorphism} if  for every $a \in \Sigma$ of arity $n$  there is an $n$-ary operation $f : A^n \to A$ in the algebra $\aalg$ such that
\begin{align*}
  h(a(t_1,\ldots,t_n)) = f (h(t_1),\ldots,h(t_n)) \qquad \mbox{for all }t_1,\ldots,t_n \in \trees \Sigma.
\end{align*}
In other words, one can index the operations in the (unindexed) algebra $\aalg$ so that it becomes a $\Sigma$-algebra and then $h$ becomes a homomorphism in the usual sense of algebras over a signature $\Sigma$, with $\trees$ seen as a $\Sigma$-algebra in the sense of Example~\ref{ex:trees-as-free-algebra}.
We say that a tree language $L \subseteq \trees \Sigma$ is \emph{recognised} by a homomorphism $h$ as above if membership $t \in L$ depends only on the value $h(a)$, i.e.~one can distinguish an accepting subset $F \subseteq A$ such that $L$ is equal to the inverse image $h^{-1}(F)$. A tree language is said to be \emph{recognised} by an algebra if it is recognised by some homomorphism into it. A tree language is called \emph{regular} if it is recognised by a homomorphism into some finite algebra. 

A homomorphism can be viewed as a deterministic bottom-up tree automaton, with the states being the universe of the algebra $\aalg$ and the transitions being defined according to the homomorphism. The only difference between such a homomorphism $h$ and a (deterministic bottom-up tree) automaton is that an automaton comes with a set of accepting states.

\subsection{Logic on trees} Regular tree languages are an important topic in formal language theory. There are many variants  (e.g.~unranked trees that appear in the study of {\sc xml} or infinite trees as studied in the theory of verification), but already there is much to say about   finite trees over a ranked alphabet, as discussed in this paper. Our main topic of interest is tree languages that can be defined using  logic, mainly monadic second-order logic and its fragments. The paradigm  dates back to results of B\"uchi, Trakhtenbrot and Elgot in the early 1960's: we view a tree (or word) as a  relational structure, and then associate to each formula of logic  those trees where the formula is true. For more on this paradigm, see~\cite{thomas1997languages}. 

 A tree $t \in \trees \Sigma$ is  interpreted as a relational  structure, in the sense of model theory,  as follows. The universe is the set of nodes. The vocabulary contains a unary predicate   $a(x)$ for every $a \in \Sigma$, which is interpreted as the nodes with label $a$, as well as the following   binary predicates: a descendant predicate, and an $i$-th child predicate for every $i \in \set{1,2,\ldots}$.  Call this structure $\underline t$. To describe properties of $t$, in terms of the structure $\underline t$, we  use the following logics, listred in decreasing order of expressive power:
\begin{itemize}
	\item \emph{Monadic second-order logic}, which quantifies over nodes and sets of nodes.
	\item \emph{Chain logic}~\cite{DBLP:conf/caap/Thomas84}, which quantifies over nodes and chains, where chains are sets of nodes that are totally ordered by the descendant relation\footnote{A natural alternative would be to consider antichain logic, where set quantification is restricted to sets that are antichains with respect to the descendant relation. In~\cite{DBLP:conf/fct/PotthoffT93} it is shown that, in the absence of letters of arity one, antichain logic has the same expressive power as full monadic second-order logic.}.
	\item \emph{First-order logic}, which quantifies over nodes.
\end{itemize}
Note that \mso and chain logic have the same syntax, but the semantics are different because of the way the second-order variables are interpreted: in \mso they range over arbitrary sets of nodes, and in chain logic they range only over chains.
A tree language is called \emph{definable} in one of the logics above if there is a formula of the  logic, over the vocabulary described above, such that a tree $t$ belongs to the language if and only if the formula is true in the structure $\underline t$. As shown by Thatcher and Wright already in the first paper on regular tree languages~\cite{thatcher}, a tree language is regular if and only if it is definable in monadic second-order logic. As already mentioned, the three logics discussed above have different expressive powers, and the strictness of the inclusions is witnessed by examples below as follows:
\begin{center}
	first-order logic $\stackrel {\text{Example~\ref{ex:unbounded-alternation}}}\subsetneq $ chain logic $\stackrel {\text{Example~\ref{ex:parity}}}\subsetneq $ \mso $\stackrel {\text{\cite{thatcher}}} = $ all regular tree languages.
\end{center}

\begin{exa}[\bf Boolean formulas with conjunction only]
Suppose that the alphabet is $\set{\lor,0,1}$ with $\lor$ having arity two, and $\set{0,1}$ having arity zero. A tree over this alphabet is a Boolean formula that only uses disjunction. For such a tree, we can talk about its value in $\set{0,1}$, which is obtained by simply evaluating the formula. The tree language consisting of Boolean formulas which are true is defined by the following formula of first-order logic
\begin{align*}
  \exists x \ 1(x)
\end{align*}
which says that some node has label 1. 	This node is necessarily a leaf, since label 1 has arity zero.
\end{exa}

\begin{exa}[\bf Boolean formulas of bounded alternation]
Let us continue the example of Boolean formulas by adding a binary symbol $\land$ to the alphabet, i.e.~the alphabet is now $\set{\lor,\land, 0,1}$. We  say that a tree is in {\sc cnf} form if a node with disjunction does not have conjunctions in its subtree, i.e.~it satisfies the following first-order sentence, which uses $x \le y$ for the descendant relation:
\begin{align*}
  \forall x \ \forall y  \ {\color{red}\lor}(x) \ \land\  x \le y\    \Rightarrow \neg ({\color{red}\land} (y)).
\end{align*}
In the above formula, we used the red colour to distinguish the unary predicate ``node $x$ has label $\lor$'' from the logical connective $\lor$, likewise for $\land$. If a tree is in {\sc cnf} form, then its value as a Boolean formula is ``true'' if and only if it satisfies the following first-order formula
\begin{align*}
  \forall x\ \big(   {\color{red}\land} (y) \ \Rightarrow \ \exists y  ( x \le y \land 1(y))\big).
\end{align*}
Using similar ideas, one can define in first-order logic the set of true formulas in {\sc dnf} formula, or more generally, the set of true formulas of any fixed alternation between $\lor$ and $\land$. 
 \end{exa}
 
 \begin{exa}[\bf  Boolean formulas of unbounded alternation]\label{ex:unbounded-alternation}
 In the previous example, we discussed true Boolean formulas with bounded alternation of $\lor$ and $\land$. When the alternation is unbounded, first-order logic is no longer sufficient to define the set of true Boolean formulas~\cite{DBLP:phd/dnb/Potthoff94}, and even chain logic is not sufficient, see Lemma 2.5.12 in~\cite{doktorat-bojanczyk}. On the other hand, \mso is sufficient, by using a formula which guesses the set $X$ of nodes which have subtrees that evaluate to true:
 \begin{align*}
  \exists X \ \forall x \quad  x \in X \ \Leftrightarrow \ \land \begin{cases}
  	{\color{red}\land} (x) \Rightarrow (\forall y \ \mathrm{child}(x,y) \Rightarrow y \in X) \\
  	{\color{red}\lor} (x) \Rightarrow (\exists y \ \mathrm{child}(x,y) \land y \in X)\\
  	\neg (0 (x))
  \end{cases}
\end{align*}
(Technically speaking, the  child relation is the disjunction of the first child relation and the second child relation.)
The same idea works for any language recognised by a finite algebra (equivalently, tree automaton), except that instead of existentially guessing one set $X$, one might need to guess more sets to represent the carrier of the algebra.
 \end{exa}

\begin{exa}[\bf Parity is not first-order definable]\label{ex:parity} Suppose that the ranked alphabet is this: \mypic{11} Every  tree over this alphabet looks like this, for some choice of $n$: \mypic{12} Let $L$ be the set of trees over this alphabet where the number of  nodes is even. Like any regular language, this language is definable in \mso. The formula  uses existential set quantification to guess those tree nodes that have an even number of nodes in their subtree. If we take the same formula, and interpret it as a formula of chain logic, then it will also define the same language. This is because when the alphabet has only symbols of arity at most one, then all sets of nodes are necessarily chains.  Therefore, the language $L$ is definable in both \mso and chain logic. (Actually, chain logic is contained in \mso with respect to expressive power, since one can easily check in \mso if a set of nodes is a chain.) First-order logic is too weak to define $L$. This can be shown using the same Ehrenfeucht-Fra\"iss\'e argument which shows that ``words of even length'' cannot be defined in first-order logic, see e.g.~Theorem IV.2.1 in~\cite{straubing}.
\end{exa}

 \subsection{The definability problem.} For a fragment of monadic-second order logic on trees, e.g.~chain logic or first-order logic, the \emph{definability problem} is the decision problem: given a regular tree language, decide if the language can be defined by some formula of the lgoic. It makes little sense to talk about the definability problem of full monadic second-order logic, since this logic recognises all regular tree languages, and therefore the algorithm would always say ``yes''.
 One way of representing the input for the algorithm is by giving a  homomorphism
 \begin{align*}
  h : \trees \Sigma \to \aalg
\end{align*}
into a finite algebra which recognises it, together with the accepting set, i.e.~the image of the language under $h$. Another representation would be a formula of \mso defining the language. As long as  decidability but not computational complexity is concerned, the choice between the two representations above (or many other) is unimportant, because there are effective conversions both ways (although the conversion from \mso to an algebra is nonelementary, see e.g.~the remarks on p.~398 of~\cite{thomas1997languages}). One of the major open problems in formal language theory is the following question, first posed by Wolfgang Thomas in~\cite{DBLP:conf/caap/Thomas84}: is definability in first-order logic decidable? There exist several different characterisations of first-order logic for trees, e.g.~using algebra~\cite{doktorat-bojanczyk,DBLP:journals/corr/abs-1208-6172, esik-weil1} or using temporal logic~\cite{DBLP:conf/caap/Thomas84}, but none of these characterisations yield algorithms for the definability problem. The main challenge is that by first-order logic, we mean first-order logic with the descendant predicate, which breaks techniques using Hanf locality; in fact definability is decidable for first-order logic with the child relations only~\cite{ben-seg}. 
 A related question, which turns out to be closer to the focus of this paper is: is definability in chain logic decidable? The definability question for chain logic was studied in~\cite{doktorat-bojanczyk,DBLP:journals/corr/abs-1208-6172}, but only non-effective characterisations were presented there.
 
The goal of this paper is to shed some light on the definability problems described above by looking at related results from universal algebra. The conclusion is going to be that there is some hope to use the structural theory of finite algebras to decide  the definability problem for chain logic; but for first-order logic the road ahead seems to be longer.

\section{Definability as an algebraic question}
\label{sec:definability-algebraic}
As mentioned in the introduction, the connection between logic and algebra is well understood for word languages. In the case of word languages, the fundamental results are: (a) the Sch\"utzenberger Theorem, which says that a word language is definable in first-order logic with order if and only if its syntactic semigroup contains no group; and (b) the Eilenberg Pseudovariety Theorem, which shows that pseudovarieties of languages are in one-to-one correspondence with pseudovarieties of semigroups.  Generalising (a) to trees is a major open problem. On the other hand, (b) lends itself much more easily to generalisations, and this has been done for the first time in~\cite{steinby79}, and then several other times, because of different notions of algebra, see the discussion on p.~29 of~\cite{Gecseg1997} or Section 4 in~\cite{DBLP:journals/corr/Bojanczyk15}. Since pseudovariety theorems can be a bit longwinded, in this section we concentrate only on one aspect of such theorems, namely sufficient conditions for a class of languages to be characterisable in terms of the syntactic algebra.

\subsection{Syntactic Algebra}

Before defining the syntactic algebra, let us begin by recalling some standard algebraic terminology. Let $\aalg$ be an algebra.  A \emph{reduct} of $\aalg$ is any algebra obtained from it by keeping the same carrier, but removing some of the operations. A \emph{subalgebra} of $\aalg$ is any algebra obtained from it by restricting the carrier to some subset that is closed under all operations in the algebra. A \emph{congruence} in $\aalg$ is an equivalence relation on the carrier which is compatible with all the operations in the usual sense; given a congruence $\sim$ one can define a quotient algebra $\aalg/_\sim$ in the usual way. We say that an algebra $\aalg$ divides an algebra $\balg$ if $\aalg$ can be obtained from $\balg$ by: first taking a reduct, then a subalgebra, and then a quotient. The following theorem is folklore, see e.g.~Proposition 11.2 in~\cite{Gecseg1997}.

\begin{thm}[Myhill-Nerode for trees]
	For every regular language $L \subseteq \trees \Sigma$ there exists a finite algebra  which recognises $L$, and furthermore divides every other  algebra recognising~$L$.
\end{thm}
\begin{proof}
	[Proof sketch] Define a \emph{context} over alphabet $\Sigma$ to be a term over $\Sigma$ with one variable $x$, such that the variable $x$ appears exactly once. If $p$ is a context, then it induces a natural function
	\begin{align*}
  [p] : \trees \Sigma \to \trees \Sigma 
\end{align*}
which maps a  tree $t$ to the  result of replacing $x$ with $t$  inside $p$.  Define a \emph{derivative} of $L$ to be any language of the form $\set{t : [p](t) \in L}$ for some context $p$. A classical argument in the style of the Myhill-Nerode theorem shows that the equivalence relation on $\trees \Sigma$ defined by
\begin{align*}
  t \sim t' \qquad\mbox{if $t,t'$ belong to the same derivatives of $L$}
\end{align*}
is a congruence in the algebra $\trees \Sigma$. The quotient under this equivalence is the algebra in the statement of the theorem.
\end{proof}
An algebra $\aalg$ as in the conclusion of the above theorem is called a \emph{syntactic algebra} for $L$. It is not difficult to see that the syntactic algebra is unique up to isomorphism, and hence it makes sense to talk about \emph{the} syntactic algebra. 

\subsection{A sufficient condition for the existence of an algebraic characterisation}

Let us return to the problem of classifying logics on trees, such as first-order logic. It would be nice if definability of a regular tree language $L$ by a logic $\langclass$ could be decided by only looking at the syntactic algebra of $L$. This is indeed the case, as long as the logic satisfies basic closure properties. The closure properties are Boolean combinations, derivatives (as defined in the proof of the Myhill-Nerode theorem), and inverse images under relabelings. A relabeling is simply an arity preserving function between alphabets $f : \Sigma \to \Gamma$, which can be lifted to trees in the obvious way. A class of languages $\langclass$ is called closed under inverse images of relabelings if for every language $L \subseteq \trees \Gamma$ in the class, and every relabeling $f : \Sigma \to \Gamma$, the inverse image $f^{-1}(L) \subseteq \trees \Sigma$ also belongs to the class.

\begin{thm}\label{thm:eilen1}
	Let $\langclass$ be a class of regular languages which is closed under Boolean combinations (including complementation\footnote{We would like to mention a slightly subtle point about complementation: technically speaking a regular tree language is a pair:   (the set of trees in the language, the  alphabet). Complementation depends on this alphabet, e.g.~complementing $\emptyset \subseteq \trees \Sigma$ depends on $\Sigma$. There exist very weak logics for which a set of trees might be definable  over one alphabet (e.g.~by the formula ``true'') but not over a bigger alphabet.}), inverse images of relabelings, and derivatives.  Then membership $L \in \langclass$ depends only on the syntactic algebra of $L$.
\end{thm}
\begin{proof}[Proof sketch]
	Suppose that $L \subseteq \trees \Sigma$ is a regular language, and let 
	\begin{align*}
  h : \trees \Sigma \to \aalg
\end{align*}
be a homomorphism into the syntactic algebra which recognises $L$. Using an adaptation of the classical proof of the Eilenberg Pseudovariety Theorem, one can show that for every subset $F$ of the universe in $\aalg$,  the inverse image $h^{-1}(F)$ is a finite  Boolean combination of derivatives of $L$ (see e.g. Lemma 4.7 in~\cite{DBLP:journals/corr/Bojanczyk15}). The homomorphism $h$ must necessarily use all operations in the algebra $\aalg$, and therefore every homomorphism
\begin{align*}
  g : \trees \Gamma \to \aalg.
\end{align*}
can be decomposed as a composition $h \circ f$ where $f$ is some relabelling $\Gamma \to \Sigma$. It follows that every language recognised  by $\aalg$ can be obtained from $L$ by taking derivatives, Boolean combinations, and inverse images of relabelings. Therefore all languages recognised by $\aalg$  are also in $\langclass$.
\end{proof}

It is not difficult to see that the class of tree languages definable in first-order logic satisfies the assumptions of Theorem~\ref{thm:eilen1}, and therefore definability in first-order logic can be decided by only looking at the syntactic algebra. The theorem, unfortunately, says nothing about what specifically is the property that we are looking for, and in particular it does not lead to an algorithm deciding if a language can be defined in first-order logic (for some artificial logics satisfying the assumptions of the theorem, definability is undecidable). The same remarks apply to chain logic.

\section{On the structure of finite algebras}
\label{sec:clones}
In the previous section we explained how problems such as ``can tree language $L$  be  defined in first-order logic?'' can be reduced to studying properties of finite algebras, namely the syntactic algebra of $L$. Such an approach was eminently successful in the study of regular languages, due to the well understood structural theory of finite semigroups. What about trees and the accompanying algebras?

 There is a rich  structural theory for finite algebras, including the famous Tame Congruence Theory of Hobby and McKenzie~\cite{hobby1988structure}.  However, this theory is little known in the formal language community. One of the main goals of this paper is to give some references about the structural theory of finite algebras that could be useful to the formal language community, and make some rudimentary observations about how that theory may or may not be applied.

\subsection{Structural theory of finite algebras}
An important step in the classification of finite algebras is to consider not just the basic operations given in an algebra, but also their compositions.  Suppose that  $\aalg$ is a finite algebra, whose set of operations is $\Sigma$. We can view $\Sigma$ as a ranked alphabet.   A term over $\Sigma$ with $n$ variables defines a function $f : A^n \to A$ in the natural way, such a function is called a \emph{term operation} in $\aalg$. For example, in the lattice algebra $(2,\lor,\land)$, the ternary majority operation is a term operation, as witnessed by the following term \mypic{13}
A \emph{polynomial operation} $A^n \to A$ is defined like a term operation, except that we are allowed to use constants for any element in the algebra (in general, such constants need not be part of the operations).  For example, if $\aalg$ is the semi-lattice $(2,\land)$ then the constant 1, seen as an operation $A^0 \to A$, is a polynomial (of arity zero) but not a term operation. 

When classifying regular tree languages, the difference between polynomials and terms is insignificant. The reason is that if we have a tree language recognised by a  homomorphism $h : \trees \Sigma \to \aalg$, then the algebra $\aalg$ contains a constant for every letter in $\Sigma$ of arity zero, and  therefore every element in the image of $h$  is described by a term.  This means that the polynomial operations and the term operations are the same, at least when restricted to the image of $h$. In this particular paper, we will be mainly talk about polynomials. We write $\pol \aalg$ for the algebra obtained from $\aalg$ by adding all polynomials to the operations. We write $\pol_n \aalg$ for the set of $n$-ary polynomials in $\aalg$. We say that two algebras $\aalg, \balg$ are \emph{polynomially equivalent} if the algebras $\pol \aalg, \pol \balg$ are isomorphic.

We present below a very brief discussion of the structural theory of finite algebras. We begin with a remarkable theorem of P\'alfy, which characterises, up to polynomial equivalence, all finite algebras satisfying a certain  condition. One of the types in the characterisation is  vector spaces over  finite fields, which are viewed as algebras in the following way: the carrier is the elements of the vector space, there is a binary operation $+$ for addition of vectors, and for every $x$ in the finite field there is a unary operation for scalar multiplication $v \mapsto x \cdot v$.
\begin{thm}[P\'alfy~\cite{PPP}]\label{thm:palfy}
	Let $\aalg$ be a finite algebra which is minimal in the following sense: every polynomial $f \in \pol_1 \aalg$ is either a constant function or a bijection of the universe. Then $\aalg$ is polynomially equivalent to an algebra of one of the following types:
	\begin{enumerate}
		\item an algebra with only unary operations;
\item a vector space over a finite field;
 	 		\item the Boolean algebra $(2,\lor,\land,\neg)$;
	 		\item the lattice $(2,\lor,\land)$;
 		\item the semi-lattice $(2,\lor)$.
	\end{enumerate}
\end{thm}
What P\'alfy actually proved is that if a finite algebra is minimal and its carrier has size at least three, then it is of type (1) or (5) above, see e.g.~Theorem 4.7 in~\cite{hobby1988structure}. Together with an analysis of two element algebras, see Lemma 4.8 in~\cite{hobby1988structure}, we get Theorem~\ref{thm:palfy}. Building on the above result,  Hobby and McKenzie developed a structural theory of finite algebras, called Tame Congruence Theory.  The starting point is that for every finite algebra, one can assign to some  pairs of  congruences (importantly, these pairs include all pairs of congruences such that one is included in the other, and there are no congruences in between) a type which is one of the five items in the P\'alfy theorem. It turns out that the analysis of the types that appear in an algebra yields a lot of information about the algebra itself; this is the subject of Tame Congruence Theory. The structural theory of finite algebras, including Tame Congruence Theory, has been very successful in the classification of Constraint Satisfaction Problems, see e.g.~the survey~\cite{DBLP:conf/dagstuhl/BartoKW17}. This raises hopes for a similar application to the classification of logics on trees, such as chain logic or first-order logic. So far, there are no such applications, but we hope that this paper might motivate cooperation between the two communities, eventually leading to  some progress. We only make here one small observation: if a language is definable in chain logic (or a weaker logic), then its syntactic algebra will only have types (1) and (5), as discussed below.

If a pair of congruences in a finite algebra $\aalg$ has type $i$, then one can find an algebra of type $i$ that divides $\aalg$. The class of algebras that recognise only languages in chain logic is closed under division, and it does not contain any algebras of types (2), (3), or (4), see~\cite{doktorat-bojanczyk}. It  follows that a necessary condition for a tree language to be definable in chain logic (in particular, in first-order logic), is that in its syntactic algebra, all congruence pairs must have type (1) or (5). This is not a sufficient condition. There exists a prime (i.e.~no nontrivial congruences) finite algebra $\aalg$ where the only congruence pair (i.e.~the identity congruence and the full congruence) has type (5), but the lattice $(2,\lor,\land)$ can be obtained from $\aalg$ as a subalgebra of a reduct, see Example 5 in~\cite{vanderfwerf}.

\subsection{First-order logic}
As mentioned above,  one is tempted to use the structural theory of finite algebras to classify regular languages, e.g.~to decide if a regular language can be defined in first-order logic.  In the following example, we show that specifically first-order logic might be a bad place to start.  The problem is that polynomial equivalence (which in our context is the same as term equivalence) is too coarse to decide membership in first-order logic.

\begin{exa}[\bf First-order logic is not a clone invariant]\label{ex:potthof} We show two regular tree languages such that: one is first-order definable,  the other is not,  but their syntactic algebras are polynomially equivalent. Consider the following ranked alphabet $\Sigma$: \mypic{3}Let $L \subseteq \trees \Sigma$ be  those trees where every leaf is at even depth.  
	The syntactic algebra of this language has three elements, ${0,1,\bot}$, with the functions corresponding to the letters being defined by 
	\begin{align*}
		\includegraphics[page=6,scale=0.2]{pics} = 0 \qquad \includegraphics[page=5,scale=0.2]{pics}(a)= \begin{cases}
		0 & \mbox{ if $a=1$}\\
		1 & \mbox{ if $a=0$}\\
		\bot & \mbox{ otherwise}
		\end{cases} \qquad
		\includegraphics[page=4,scale=0.2]{pics}(a,b) = \begin{cases}
			0 & \mbox{ if $a=b=1$}\\
			1 & \mbox{ if $a=b=0$}\\
			\bot & \mbox{ otherwise}
		\end{cases}
	\end{align*}
	The language $L$ is not first-order definable, for the same reasons as discussed in Example~\ref{ex:parity}, i.e.~because a formula of first-order logic cannot distinguish between the following trees for large enough $n$: \mypic{7} Define the language $K$ to be the same as $L$, except that the arity one symbol symbol  \includegraphics[page=5,scale=0.2]{pics}  is dropped from the alphabet. A surprising result by Potthoff~\cite{DBLP:phd/dnb/Potthoff94} is that the  language $K$ is first-order definable, see also page 3 in~\cite{bojanczyk-tree-algs}.  The syntactic algebra for $K$ is the same as for $L$, except that it is missing the operation corresponding to the dropped letter  \includegraphics[page=5,scale=0.2]{pics}. Nevertheless, these two algebras are polynomially equivalent (in fact, term equivalent), because  we have: \mypic{8}
\end{exa}

\subsection{Polynomial language pseudovarieties}
The problem witnessed by Example~\ref{ex:potthof} is that the  class of first-order definable tree languages is not closed under inverse images of tree homomorphisms, as described below. A function
\begin{align*}
  h : \trees \Sigma \to \trees \Gamma
\end{align*}
is called a \emph{tree homomorphism} if for every letter $a \in \Sigma$ there is some term $t_a$ over $\Gamma$ of same arity as $a$, such that $h(t)$ is obtained by replacing each letter by the corresponding term. For example, consider the homomorphism which which is defined by the following family of terms \mypic{9}
If we apply the above homomorphism to a tree without binary branching, then the result is a balanced binary tree of same depth as the input, as illustrated below \mypic{10}
It is not difficult to see that the language $L$ in Example~\ref{ex:potthof} is the inverse image, under the above homomorphism, of the language $K$ in the same example. Since $K$ is definable in first-order logic and $L$ is not, it follows that first-order logic is not closed under inverse images of tree homomorphisms\footnote{Since we already have the picture, we can explain the intuition why $K$ is first-order definable.   The main observation is the following. A balanced binary tree has all nodes (equivalently, some node) at even depth if and only if it satisfies the following property, which can be defined in first-order logic: there exists a leaf $x$ which is a first child and such such that the sequence of child numbers on the path from the root to $x$ is of the form: first child, second child, first child, second child, etc.}.

The following theorem shows that inverse images under tree homomorphisms are almost all that is necessary for being able to characterise a class of languages purely by properties of its syntactic algebra up to polynomial equivalence. (Recall that for syntactic algebras of tree languages, polynomial operations are already term operations, so term equivalence could be used in the theorem as well.)

\begin{thm} 
Let $\langclass$ be a class of regular tree languages which is closed under Boolean combinations (including complementation), inverse images of tree homomorphisms, and derivatives.  Then membership $L \in \langclass$ depends only on $\pol \aalg$ where $\aalg$ is the syntactic algebra of $L$.
\end{thm}

Let us use the name \emph{polynomial language pseudovariety} for a class of regular tree languages which satisfies the assumptions of the above theorem. As we have seen in Example~\ref{ex:potthof}, the class of first-order definable tree languages is not a polynomial language pseudovariety, which means that one cannot study first-order logic on trees purely in terms of polynomial operations. One example of a polynomial language pseudovarieties is chain logic, which can be proved using a suitably defined Ehrenfeucht-Fra\"iss\'e game. Here is another example.

\begin{exa}[\bf Path languages]\label{ex:path}
For a ranked alpahbet $\Sigma$, define $[\Sigma]$ to the set
\begin{align*}
  \set{a \in \Sigma : \mbox{$a$ has arity 0}} \cup \set{(a,i) : \mbox{$a$ has arity $n \ge 1$ and $i \in \set{1,\ldots,n}$}}
\end{align*}
A root-to-leaf path $\pi$ in a tree $t \in \trees \Sigma$ can be interpreted as a word $[t,\pi]$ over the alphabet $\Sigma$ according to the following picture: \mypic{19}
For a word language $L \subseteq [\Sigma]^*$, define  $\mathsf A L$ to be the set of trees $t \in \trees \Sigma$ such that the labelling $[t,\pi]$  of every root-to-leaf $\pi$ path belongs to $L$. A language of the form $\mathsf A L$ for $L$ a regular word language is called a \emph{universal path language}.  Universal path languages are exactly the tree languages recognised by deterministic top-down tree automata, see e.g. Section 1.6 in~\cite{tata2007}. A tree language $L$ is universal if and only if it is equal to
\begin{align*}
	\mathsf A \set{[t,\pi] : \mbox{$\pi$ is a root-to-leaf path in some $t \in L$}},
\end{align*}
in particular one can decide -- using an equality check on tree automata -- if a tree language is universal.   Define a \emph{path language} to be any tree language which is a Boolean combination of universal path languages. One can show that path languages form a polynomial language pseudovariety (see the discussion after Theorem~\ref{thm:path-dtop}). It is an open problem whether membership in this variety is decidable, see e.g.~page 27 of~\cite{doktorat-bojanczyk}.
\end{exa}

We conjecture that membership is decidable in the two polynomial language pseudovarieties described above, chain logic  and path languages, and that methods of universal algebra could be useful for this.

\section{Transducers and the matrix power}
\label{sec:matrix}
In this section, we discuss the connection between an algebraic concept (the matrix power) and a machine model (deterministic top-down transducers). We show that  these two are essentially the same thing. One corollary of this equivalence is the following  characterisation of path languages as discussed in Example~\ref{ex:path}: a tree language is a path language if and only if  it is  recognised by some matrix power of the semi-lattice $(2,\land)$, see Theorem~\ref{thm:path-dtop}. The proofs in this section are essentially  syntactic rewritings of one definition into another, and require no combinatorial insights.

\subsection{Matrix power.} The matrix power is an operation which generalises the standard (Cartesian) power of an algebra. The presentation for matrix power that we use here is based on~\cite{szendrei1990simple}, for a discussion on the history of this operation see~\cite{Taylor1975}. Let  $\aalg$ be an algebra and let $n \in \set{1,2,\ldots}$. Define the $n$-th matrix power of $\aalg$, denoted by $\aalg^{[n]}$, to be the following algebra with carrier $A^n$.  For every $k \in \set{0,1,\ldots}$ and for every tuple
\begin{align*}
  f_1,\ldots,f_n \in \pol_{n \cdot k} \aalg
\end{align*}
 of  polynomial operations  in $\aalg$, each one of arity $n \cdot k$, the matrix power  contains a $k$-ary operation defined by
\begin{align*}
	\bar a_1,\ldots, \bar a_k \in A^n \qquad \mapsto \qquad (f_1(\bar a_1,\ldots,\bar a_k),\ldots,f_1(\bar a_1,\ldots,\bar a_k)).
\end{align*}
Note how the definition depends only on the polynomials of the algebra, and hence polynomially equivalent algebras will have the same matrix powers.
In this paper we will mostly be interested in matrix powers of the semi-lattice, as discussed in the following example.

\begin{exa}[\bf Matrix powers of the semi-lattice]
	Consider  the semi-lattice $(2,\land)$. An $n$-ary operation in this algebra is either a constant in $\set{0,1}$ or a conjunction of some subset $I \subseteq \set{1,\ldots,n}$ of its arguments. A $k$-ary operation in the $n$-th matrix power of $(2,\land)$ is a tuple of such operations, each one with arguments $n \cdot k$. 
An operation in the matrix power can be viewed as a type of circuit, as in the following picture for $n=3$ and $k=2$:  \mypic{16} If we would be using matrix products of the Boolean algebra $(2,\lor,\land,\neg)$, then the operations in the matrix power would correspond to general Boolean circuits, i.e.~ones which can use all Boolean operations instead of $\land$ only.
Terms in the matrix power correspond to tree-shaped circuits as in the following picture, which shows a term with zero variables, and hence only output values: \mypic{15} 
In Theorem~\ref{thm:path-dtop} we will show that a tree language is  recognised by a homomorphism
	\begin{align*}
  h : \trees \Sigma \to (2,\land)^{[n]}.
\end{align*}
if and only if it is a path language in the sense of Example~\ref{ex:path}. 

As mentioned above, adding $\lor$ and $\neg$ to the algebra would allow use to model arbitrary Boolean circuits. In fact, this extension would allow us to capture all finite algebras in the following sense:
every finite algebra is isomorphic to a subalgebra of a reduct of some matrix power of the Boolean algebra $(2,\land,\lor,\neg)$. The idea is to encode each element of an algebra as a bit vector and use circuits to compute the values in the algebra. We do not even need negation, if we use an encoding  that produces  bit vectors with only one coordinate being true.
\end{exa}

	\subsection{Transducers.} The matrix power is intimately connected with an operation on trees which is called a \emph{deterministic top-down transducer} (\dtop). A \dtop can be viewed as a generalisation of a tree homomorphism which allows control states; conversely a tree homomorphism is the same thing as a \dtop with only one control state. 	
The syntax of a \dtop consists of the following ingredients:
	\begin{itemize}
		\item two ranked alphabets $\Sigma,\Gamma$, called the \emph{input} and \emph{output} alphabets;
				\item a finite set $Q$ of \emph{states}, together with an \emph{initial state} $q_0 \in Q$;
				\item for each letter $a \in \Sigma$ of arity $n$, a \emph{transition function}
				\begin{align*}
  \delta_a : Q \to \mathsf T_\Gamma (Q \times \set{1,\ldots,k}),
\end{align*}
where $\mathsf T_\Gamma X$ represents terms over alphabet $\Gamma$ with variables $X$.
	\end{itemize} We would like to underline that the terms produced by the transition functions do not need to use all their variables. 
	
	We now describe the semantics of a \dtop.	For each state  $q \in Q$, we define a function
\begin{align*}
f_q : \trees \Sigma \to \trees \Gamma.
\end{align*}
The definition  is by mutually recursive induction on the size of the input tree. The function $f_q$ maps a tree $a(t_1,\ldots,t_n)$ to the tree obtained from taking the term $\delta_a(q)$, which uses variables from the set $Q \times \set{1,\ldots,k}$,   and then applying the  substitution which maps variable $(p,i)$  to the tree $f_p(t_i)$ obtained from induction. 
The semantics of a \dtop is defined to be the  function $f_{q_0}$ corresponding to the initial state.  
By abuse of notation, we do not distinguish between the transducer (i.e.~its syntax) and the function that it defines (i.e.~its semantics).

The following result shows the connection between matrix power and \dtops. To the authors' best knowledge, this connection was not observed before.   
 
\begin{thm}\label{thm:matrix-power}
	Let $\aalg$ be a finite algebra, where each element of the carrier is represented by a constant. The following  conditions are equivalent for every tree language $L \subseteq \trees \Sigma$:
	\begin{itemize}
				\item $L$ is recognised by a matrix power of $\aalg$;
		\item $L$ is a Boolean combination of languages of the form
		\begin{align*}
  f^{-1}(K) \quad \mbox{ for some  \dtop $f$ and some $K$  recognised by $\aalg$.}
\end{align*}
	\end{itemize}
\end{thm}
\begin{proof} The proof of this theorem is simply by unfolding the definitions.
	Let us begin with the top down implication. Consider a homomorphism
	\begin{align*}
  h : \trees \Sigma \to \aalg^{[n]}.
\end{align*}
Let  $\Gamma$ be the set of operations in the algebra $\aalg$, including one constant per element of its carrier, and  consider the homomorphism
\begin{align*}
\_^\aalg : \trees \Gamma \to \aalg 
\end{align*}
which inputs a tree built out of operations and simply evaluates them bottom up.  

\begin{lem}
For every $i \in \set{1,\ldots,n}$ there is a \dtop $f_i$ which makes the following diagram commute:
	\begin{align}\label{eq:com-dtop}
  \xymatrix{ \trees \Sigma \ar[d]_{f_i}\ar[r]^{h}& A^n \ar[d]^{\text{projection to $i$-th coordinate}} \\
  \trees \Gamma \ar[r]_{\_^\aalg} & A}
\end{align}
	\end{lem}
\begin{proof}[Proof of the lemma]
The only dependence of $f_i$ on $i$ is the choice of initial state, otherwise the \dtop\ \!s $f_1,\ldots,f_n$ are the same. The states of the \dtop $f_i$ are $\set{1,\ldots,n}$ and the initial state is $i$. The input alphabet is $\Sigma$ and the output alphabet is $\Gamma$. For a letter $a$ of arity $k$, the  transition relation $\delta_a$ of the \dtop maps a state $i \in \set{1,\ldots,n}$ to the $i$-th polynomial in the $n$-tuple of $(n \cdot k)$-ary polynomials which define the operation of $\aalg^{[m]}$ that corresponds to the  letter $a$ under the homomorphism $h$. By induction on the depth of a tree $t \in \trees \Sigma$, one shows that if we apply to $t$ the functions corresponding to the two paths in the diagram from the statement of the lemma (i.e.~right-down or down-right), then the resulting values are the same.  This completes the proof of the lemma.
\end{proof}
Using the above lemma, we complete the proof of  the top-down implication in the theorem. By the lemma,  for every $a \in A$ and every $i \in \set{1,\ldots,n}$, the set
\begin{align}\label{eq:slice}
  \set{ t \in \trees{ \Sigma} : h(t) \mbox{ has $a$ on coordinate $i$}}
\end{align}
is the inverse image, under  some \dtop, of some language recognised by $\alg$. This  completes the top-down implication in the theorem, because every language recognised by $h$ is a Boolean combination of languages of the form~\eqref{eq:slice}.  

To  prove the bottom-up implication in the theorem, we use the following lemma.
\begin{lem}
Let $f : \trees \Sigma \to \trees \Gamma$ be a \dtop, and let  $g : \trees \Gamma \to \aalg$ be a homomorphism. Every language  recognised by $g \circ f$  is recognised by some  homomorphism from $\trees \Sigma$ into some matrix power of $\aalg$.	
\end{lem}
\begin{proof} We assume without loss of generality that the states of the \dtop recognising $f$ are numbers $\set{1,\ldots,n}$. Consider a homomorphism
\begin{align*}
  h : \trees \Sigma \to \aalg^{[n]}
\end{align*}
defined as follows. A letter $a$ of arity $k$ is mapped to the tuple of polynomials 
\begin{align*}
  ( p_1,\ldots,p_n) \in \pol_{n \cdot k} \aalg
\end{align*}
such that $p_i$ is the result of applying the transition  function $\delta_a$  of $f$ to the state $i$ and then evaluating the resulting term over $\Gamma$ in the algebra $\aalg$ via the homomorphism $g$. By induction on the size of $t$  we show that
\begin{align*}
  h(t) = (g(f_1(t)),\ldots,g(f_n(t))) \qquad \mbox{for every $t \in \trees{\Sigma}$},
\end{align*}
where $f_i$ is the \dtop obtained from $i$ by changing the initial state to $i$. In particular, membership of a tree $t$ in any language recognised by $g \circ f$ can be determined by looking at the coordinate of $h(t)$ which corresponds to the initial state of $f$. This completes the proof  the lemma.
\end{proof}
The above lemma shows that every language $f^{-1}(K)$ as in the statement of the theorem is recognised by some matrix power of $\aalg$. To complete the proof of the bottom-up implication in the theorem, we observe that the set of languages
\begin{align*}
  \set{ L \subseteq \trees \Sigma : \mbox{$L$ is recognised by some homomorphism into some matrix power $\aalg^{[n]}$}}
\end{align*}
is  closed under Boolean operations (the idea is that the matrix  power generalises the Cartesian power). 
	\end{proof}

Following~\cite{vanderfwerf}, see the Definition on p.~40, for a class of algebras $\algclass$ define $\matrclos \algclass$ to be the class of matrix powers of algebras from $\algclass$. Inverse images under \dtops commute with Boolean operations, i.e.~if $f$ is a \dtop, in fact any function, then 
\begin{align*}
  f^{-1}(K \cup L) = f^{-1}(K) \cup f^{-1}(L),
\end{align*}
likewise for other Boolean operations. Combining this observation with  Theorem~\ref{thm:matrix-power}, we see that  if $\mathscr A$ is a class of algebras, then the  languages  recognised  by algebras from $\matrclos \algclass$ are exactly the smallest class of languages that contains all languages recognised by algebras from $\matrclos \algclass$ and which is closed under Boolean operations and inverse images under \dtops.

\subsection{Path languages and matrix power}
We now apply Theorem~\ref{thm:matrix-power} to give a characterisation of the path languages that were discussed in Example~\ref{ex:path}. The characterisation below is not effective, in the sense that it does not give an algorithm which decides if a tree language is a path language. Our hope, however, is that drawing the connection between path languages and algebra will make it easier to eventually find an effective characterisation of path languages.
\begin{thm}\label{thm:path-dtop}
	A tree language is recognised by an algebra from $\matrclos \set{(2,\land)}$ (i.e.~by a matrix power of the semi-lattice) if and only if it is a path language (equivalently, a Boolean combination of languages recognised by deterministic top-down tree automata).
\end{thm}
\begin{proof} Note that under the matrix power of an algebra $\aalg$ depends only on its polynomials, and therefore for every $n$, the $n$-th matrix power is the same for $(2,\land)$ as for $(2,\land,0,1)$. Therefore, from  Theorem~\ref{thm:matrix-power} it follows that a language is recognised by a matrix power of $(2,\land)$ if and only if it is a Boolean combination of  inverse images under \dtops of languages recognised by  $(2,\land,0,1)$. From this it is not difficult to see that a language is recognised by some matrix power of $(2,\land)$ if and only if it is a Boolean combination of languages of the form
\begin{align}\label{eq:booleval-lang}
(h \circ f)^{-1}(1) \qquad \mbox{for some } \underbrace {f : \trees \Sigma \to \trees \set{\land,0,1}}_{\text{\dtop}} ,
\end{align}
where $h : \trees \set{\land,0,1} \to \set{0,1}$ is the function which evaluated a Boolean expression. To complete the proof of the theorem, we use the following lemma, whose straightforward proof is simply an unfolding of the definitions, and is omitted here.
\begin{lem}
A tree language is of the form~\eqref{eq:booleval-lang} if and only if it is a universal path language (equivalently, it is recognised by a deterministic top-down tree automaton).
\end{lem}
\end{proof}

One corollary of the above theorem, and the Myhill-Nerode theorem, is that a tree language is a path language if and only if its syntactic algebra divides a matrix power of the semi-lattice. Deciding the latter property, when given a finite algebra, is an open problem to the authors' best knowledge. 

Another corollary of the above theorem is that path languages are closed under inverse images of tree homomorphisms, which is the harder part of showing that path languages form a polynomial language variety. The reason is that for every class of algebras $\algclass$ which is closed under polynomial equivalence (such as matrix powers of the semi-lattice), the class of tree languages recognised by algebras from $\algclass$ is closed under inverse images of tree homomorphisms.

\begin{exa}[\bf Doubly deterministic tree languages]
Let $\algclass$ be the class of algebras which have only unary operations. One can show that a tree language is recognised by an algebra from $\matrclos \algclass$ if and only if it is \emph{doubly deterministic} in the following sense: both the language and its complement are recognised by deterministic top-down tree automata (equivalently, the language and its complement are both universal path languages). Examples of doubly deterministic tree languages include ``the root label is $a$'', or ``the left most leaf is $c$''. It is decidable if a tree language is recognised by a deterministic top-down tree automaton, and therefore it is decidable if a tree language is doubly deterministic. The class of algebras $\matrclos \algclass$ was studied in~\cite{szendrei1990simple}.
\end{exa}

\section{Language nesting and the wreath product}
\label{sec:wreath}
In the previous section, we discussed connections between the matrix power and \dtops. In this section, we discuss another  connection of this type: the wreath product of algebras corresponds to nesting of tree languages. This connection is  known in the logic and  automata community, where sometimes the name \emph{cascade product} is used instead of wreath product. For word languages, the connection between wreath product and nesting  dates back to  the folklore observation that wreath products of transformation semigroup $U_2$, as in the Krohn-Rhodes Theorem, have the same recognising power as formulas of linear temporal logic. The generalisation to trees, as discussed in this section, has been observed in~\cite{doktorat-bojanczyk,DBLP:journals/corr/abs-1208-6172,ESIK2006136}.

For algebras $\aalg$ and $\balg$, define their \emph{wreath product} $\aalg \circ \balg$ to be the following algebra. Its carrier is the Cartesian product $A \times B$ of the carriers in the underlying algebras. 
The operations in the wreath product correspond to pairs $(\alpha,f)$ such that  $f$ is an operation  in $\balg$, and $\alpha$ is a  function from $B$ to operations in $\aalg$ of the same arity as $f$. If $f$ has arity $n$, then the operation corresponding to a pair $(\alpha,f)$ is the following $n$-ary operation:
\begin{align*}
  (a_1,b_1),\ldots,(a_n,b_n) \qquad \mapsto \qquad (a,b) \mbox{ where } \begin{cases}
  	b= f(b_1,\ldots,b_n) \\
  	a= (\alpha(b))(a_1,\ldots,a_n)
  \end{cases}
\end{align*}	
This completes the definition of the wreath product. 

To get a feeling for the wreath product, consider the following characterisation, which roughly corresponds to Straubing's \emph{wreath product principle}~\cite{STRAUBING1979305}.
For a ranked alphabet $\Sigma$ and an (unranked) set $X$, define $\Sigma \times X$ to be the ranked alphabet obtained by taking the Cartesian product, and  inheriting the arity information from $\Sigma$. For homomorphisms
\begin{align}\label{eq:seq-comp}
h : \trees \Sigma \to \balg \qquad   g : \trees (\Sigma \times B) \to \aalg 
\end{align}
define their \emph{sequential composition}  to be the function
\begin{align*}
  t \in \trees \Sigma \qquad \mapsto \qquad  (g(t^h),h(t)) \in A \times B
\end{align*}
where the tree $t^h \in \trees (\Sigma \times B)$ is obtained from $t$ by labelling each node with the pair (label in $t$, value under $h$ of the subtree).  The following observation shows that wreath product is essentially the same thing as sequential composition of homomorphisms.
\begin{lem}\label{lem:wreath-principle}
A  function
  $f : \trees \Sigma \to A \times B$
is a homomorphism into $\aalg \circ \balg$ if and only if it is equal to a sequential composition of some homomorphisms as in~\eqref{eq:seq-comp}.\end{lem}

 The proof of the observation is left to the reader; it follows the same lines as Theorem 4.2 in~\cite{DBLP:journals/corr/abs-1208-6172}. 
 
Let us further restate the correspondence between wreath product and sequential composition, only this time using the terminology of nesting tree languages. The idea of nesting tree languages (or words languages) comes from the study of abstract operators in temporal logics, see e.g.~\cite{DBLP:journals/logcom/BeauquierR02} for the word case or~\cite{ESIK2006136} for the tree case. To define nesting of tree languages, consider the operation 
\begin{align*}
(L_1,\ldots,L_n \subseteq \trees \Sigma,
  t \in \trees \Sigma) \qquad \mapsto \qquad t^{L_1,\ldots,L_n} \in \trees (\Sigma \times 2^n)
\end{align*}
which simply extends the label of each node $v$ of $t$ by the bit-vector indicating which of the languages among $L_1,\ldots,L_n$ contain the subtree of $t$ rooted in $v$. We say that a class of tree languages $\langclass$ is \emph{closed under nesting} if for every tree languages
\begin{align*}
L_1,\ldots,L_n \subseteq \trees \Sigma \qquad L \subseteq \trees (\Sigma \times 2^n)
\end{align*}
that are in the class $\langclass$, also the following language is in $\langclass$:
\begin{align*}
  \set{ t  \in \trees \Sigma :  t^{L_1,\ldots,L_n} \in L}.
\end{align*}

The following theorem is yet another variant of the wreath product principle.

\begin{thm}\label{thm:wreath}
	Let $\algclass$ be a class of finite algebras, and let $\mathsf{W} \algclass$ be the least class of finite algebras which contains $\algclass$ and is closed under wreath products. Then the class of languages recognised by algebras from $\mathsf{W} \algclass$ is the smallest class of languages that is closed under nesting, and contains all languages recognised by algebras from $\algclass$.
	\end{thm}
	
The proof of the above theorem is a relatively straightforward  application of Lemma~\ref{lem:wreath-principle}, and hence we omit it here. The  proof is similar to Theorem 31 in~\cite{ESIK2006136}, or Theorems 2.5.7 and 2.5.9 in~\cite{doktorat-bojanczyk}, which use the terminology of \emph{cascade product} instead of wreath product. An unranked version of the  theorem can also be found in Corollary 5.1 of~\cite{DBLP:journals/corr/abs-1208-6172}. In the following two examples, we show wreath product characterisations of two logics on trees, namely chain logic and a variant of the temporal logic {\sc ctl}.

\begin{exa}[\bf Chain logic as a class of algebras]\label{ex:chain-mw}
This example is a small variation on Theorem 2.5.9 in~\cite{doktorat-bojanczyk}. Consider the class 
 \begin{align*}
  \algclass \eqdef \wreaclos \matrclos  \set{(2,\land)}
\end{align*}
i.e.~algebras which are wreath products of matrix powers of the semi-lattice. By Lemma 46 in~\cite{vanderfwerf}, this is the same as matrix powers of wreath products of the semi-lattice, and in particular the class $\algclass$ is closed under both wreath products and matrix powers.
  By Theorems~\ref{thm:path-dtop} and~\ref{thm:wreath}, the tree  languages recognised by  algebras from $\algclass$ are exactly the closure under nesting of languages recognised by  deterministic top-down tree automata. (Boolean combinations are superfluous, since nesting can simulate Boolean combinations.) By Theorem 2.5.9  in~\cite{doktorat-bojanczyk}, this class of languages is exactly the languages definable in chain logic. From the Myhill-Nerode theorem for trees, it follows that a language is definable in chain logic if and only if its syntactic algebra divides some algebra from $\algclass$. Summing up, deciding definability of a tree language in chain logic reduces to deciding if a finite algebra divides some algebra from $\algclass$.  The class $\algclass$ and  its divisors were studied in~\cite{vanderfwerf}, but unfortunately there is still no known algorithm for deciding if an algebra divides some algebra from $\algclass$.
  \end{exa}

 \begin{exa}[\bf Direction sensitive CTL] This example is a small variation on the results from~\cite{DBLP:journals/fuin/EsikI08,DBLP:journals/fuin/EsikI08a,ESIK2006136}.
Consider the class 
 \begin{align*}
  \algclass \eqdef \wreaclos  \set{(2,\land)}
\end{align*}
i.e.~algebras which are wreath products of the semi-lattice.  We will show that tree languages recognised by algebras from $\algclass$ are exactly those which can be  defined in a certain variant of the temporal logic {\sc ctl} described below. Suppose that $\Sigma$ is a ranked alphabet, and let 
\begin{align*}
  X \subseteq \set{(a,i) : a \in \Sigma \mbox{ has rank $n \ge 1$ and $i \in \set{1,\ldots,n}$}} \qquad Y \subseteq \Sigma.
\end{align*}
Define $X \ \mathsf{until}\  Y$ to be the set of  trees $t \in \trees \Sigma$ which contain at least one node $v$ satisfying: (a) the label of $v$ is $Y$; (b) if $w$ is a proper ancestor of $v$ with label $b$, and $v$ is in the $i$-th subtree of $w$, then $(b,i) \in X$. Let use the name \emph{direction sensitive until language} for a language of the form $X \ \mathsf{until}\ Y$, and let \emph{direction sensitive {\sc ctl}} be the nesting closure of direction sensitive until languages. The name is so chosen because direction sensitive {\sc ctl} is essentially the same thing as  {\sc ctl} (without the next modality $\mathsf X$) with the difference that, unlike in standard {\sc ctl}, the until operator  is sensitive to child numbers. It is not difficult to see that  direction sensitive until languages and their complements are exactly the tree languages recognised by the semi-lattice.  Therefore, from  Theorem~\ref{thm:wreath} it follows that  a tree language is recognised by an algebra in $\algclass$ if and only if it can be defined in direction sensitive {\sc ctl}.  Furthermore, deciding if a tree language is definable in direction sensitive {\sc ctl} reduces to testing if its syntactic algebra divides some algebra in $\algclass$. As was the case in Example~\ref{ex:chain-mw}, the class $\algclass$ and divisors were studied in~\cite{vanderfwerf}, but without giving an algorithm  for the above mentioned decision problem.
 \end{exa}

\bibliographystyle{plain}
\bibliography{bib}

\end{document}